\documentclass[1p,times]{elsarticle}
\usepackage{amsmath,amsthm,amssymb,amsfonts,ascmac}
\usepackage{bm,empheq,mathtools,latexsym,color,rotating,tabularx,booktabs,tabulary,adjustbox}
\usepackage{multirow}

\newtheorem{theorem}{Theorem}[section]
\newtheorem{lemma}[theorem]{Lemma}

\theoremstyle{definition}

\newtheorem{remark}[theorem]{Remark}

\numberwithin{equation}{section}

\journal{Journal of Physics A: Mathematical and Theoretical}

\begin{document}

\begin{frontmatter}

\title{Dunkl-Supersymmetric orthogonal functions associated with classical orthogonal polynomials}
\tnotetext[label1]{The research of ST is supported by JSPS KAKENHI (Grant Numbers 16K13761, 19H01792) and that of LV by a discovery grant of the Natural Sciences and Engineering Research Council (NSERC) of Canada. The work of AZ is
supported by the National Science Foundation of China (Grant No.11771015).}

%
%

\author[Kyoto]{Yu Luo}
\ead{luo.yu.68e@st.kyoto-u.ac.jp}

\author[Kyoto]{Satoshi Tsujimoto}
\ead{tsujimoto.satoshi.5s@kyoto-u.jp}

\author[LV]{Luc Vinet}

\author[ZA]{Alexei Zhedanov}

\address[Kyoto]{Department of Applied Mathematics and Physics, Graduate School of Informatics, Kyoto University, Yoshida-Honmachi, Sakyo-ku, Kyoto 606-8501, Japan}

\address[LV]{Centre de recherches, math\'ematiques, Universit\'e de Montr\'eal, P.O. Box 6128, Centre-ville Station,
Montr\'eal (Qu\'ebec), H3C 3J7}

\address[ZA]{Department of Mathematics, School of Information, Renmin University of China, Beijing 100872,CHINA}

\begin{abstract}
We consider the eigenvalue problem associated with the Dunkl-type differential operator 
(in which the reflection operator $R$ is involved)
\begin{displaymath}
  \mathcal{L}=\partial_x R+v(x), \quad (v(-x)=-v(x)),
\end{displaymath}
in the context of supersymmetric quantum mechanical models. 
By solving this eigenvalue problem with the help of known exactly solvable potentials, 
we construct several classes of functions satisfying certain orthogonality relations. 
We call them the Dunkl-supersymmetric (Dunkl-SUSY) orthogonal functions. 
These functions can be expressed in terms of the classical orthogonal polynomials (COPs). 
The key feature of these functions is that they appear by pairs, 
i.e., $Q_n(x)$ and $Q_{n}(-x)$ are both the eigenfunctions of $\mathcal{L}$. 
A general formulation of the Dunkl-SUSY orthogonal polynomials is also presented. 

\end{abstract}

\begin{keyword}
Supersymmetric quantum mechanics \sep Dunkl-type differential operator \sep 
Classical orthogonal polynomials

\MSC[2010] primary 33C45 \sep 33C47 \sep secondary 42C05
\end{keyword} 

\end{frontmatter}

\section{Introduction}
\par
Supersymmetric quantum mechanics (SUSY QM)  
has been useful in the study of exactly solvable quantum mechanical models 
\cite{SUSYQM, DbSUSY}. 
In \cite{SUSYQMR}, the authors presented supersymmetric quantum mechanical models in 
one dimension involving differential operators of Dunkl-type (see also \cite{Plyu94,Plyu96}).
In this realization, 
the reflection operator appears in both the supersymmetric Hamiltonian and the supercharge.
The wave functions for two such systems have been obtained in \cite{SUSYQMR} and seen to define orthogonal polynomials that are themselves expressed in terms of Hermite and little -1 Jacobi polynomials respectively. We here propose to pursue the exploration of the orthogonal polynomials that occur as eigenfunctions of such Dunkl supercharges,
specifically of the operator 
\begin{equation}
\label{calL}
\mathcal{L}=\partial_x R+v(x), 
\end{equation}
where $v(-x) = -v(x)$.
%

\par

\par
A Hamiltonian $H$ is said to be supersymmetric if there are supercharges $Q$, $Q^{\dagger}$ 
such that the following superalgebra relations are satisfied 
\begin{equation}
\label{SUSY_algbra}
[Q, H]=0, \quad [Q^{\dagger}, H]=0, \quad
H=\{Q, Q^{\dagger}\}.
\end{equation}
The brackets $[,]$ and $\{,\}$ are called the commutator and the anticommutator, respectively: 
$$
[A, B]=AB-BA, \quad 
\{A, B\}=AB+BA.
$$
If the supercharge $Q$ is self-adjoint, i.e., $Q^{\dagger}=Q$, 
it follows from (\ref{SUSY_algbra}) that the $H=2Q^2$, 
and 
the model is said to be $N=1/2$ supersymmetric. 
Realizations of $N=1/2$ supersymmetric systems have been obtained in \cite{SUSYQMR}
by taking the supercharge as the following Dunkl-type differential operator: 
\begin{equation}
\label{D_supercharge}
Q=2^{-\frac{1}{2}}(\partial_x R+u(x)R+v(x)),
\end{equation}
where $u(x)$ is even, $v(x)$ is odd (i.e., $u(-x)=u(x)$, $v(-x)=-v(x)$), 
and the operator $R$ is the reflection operator which acts on $x$ as 
$Rf(x)=f(-x)$.
It is clear that $Q$ is self-adjoint, $Q^{\dagger}=Q$, 
and the Hamiltonian $H$ is then
\begin{equation}
\label{D_Ham}
H=-\partial^2_x+(u^2(x)+v^2(x)+u'(x))-v'(x)R. 
\end{equation}
Notice that if $u(x)=0$, then the supercharge $Q$ and the Hamiltonian $H$ become 
\begin{equation}
\label{QH}
Q=2^{-\frac{1}{2}}(\partial_x R+v(x)), \quad 
H=-\partial^2_x+v^2(x)-v'(x)R,
\end{equation}
and they satisfy
\begin{equation}
\label{QH1}
\{Q, R\}=0, \quad [H, R]=0.
\end{equation}
The assumption that $v(x)$ is odd is essential for the relations (\ref{QH1}) to be achieved. 
Consequently, these relations together with (\ref{SUSY_algbra}) imply that, for $Q$ and $H$ given in (\ref{QH}): 
\begin{enumerate}
\item[(a)] the operators $Q$ and $H$ share the same eigenfunctions: $Q\psi(x)=\nu\psi(x)$, $H\psi(x)=E\psi(x)$, 
where $E=2\nu^2$; 
\item[(b)] their eigenfunctions appear in pairs $\psi(x), \psi(-x)$: 
$$
Q\psi(x)=\nu\psi(x), \quad
Q\psi(-x)=-\nu\psi(-x);
$$
\begin{equation}
\label{H_eq}
H\psi(x)=E\psi(x), \quad
H\psi(-x)=E\psi(-x).
\end{equation}
\end{enumerate}

\par 
We shall focus on models described by (\ref{QH}) in the remainder of this paper 
and will assume that the operator $\mathcal{L}$ is non-degenerate, 
i.e., all that eigenvalues of $\mathcal{L}$ are distinct. 
For results on the eigenvalue problem related with the most general first-order Dunkl-type differential operator 
$$
L=F_{0}(x)+F_{1}(x)R+G_{0}(x)\partial_x+G_{1}(x)\partial_x R
$$
with arbitrary functions $F_{0}(x)$, $F_{1}(x)$, $G_{0}(x)$, $G_{1}(x)$ 
one can refer to \cite{BoThmDkl}. 

\par

\par
The rest of this paper will unfold as follows. 

\par
In Section 2, we determine the general properties of the eigenfunctions of an operator such as 
$H$ that commute with the reflection operator. 
The connection with supersymmetric quantum mechanics will be made manifest at the same time. 
Concentrating in Section 3 on the special Dunkl operators ${\mathcal L}$, we shall exploit the exact solvability of certain 
shape-invariant potentials to obtain explicitly a number polynomial families arising in the wavefunctions. 
These families will be presented in Section 4 without reference to the SUSY models from which they originated and 
the orthogonality relations will be specified. The essential properties of Dunkl-supersymmetric OPs will then 
be identified from the examples thus collected so as to provide in Section 5 a general characterization of 
these polynomials in terms of sets of symmetric polynomials as well as their recurrence relations. 
A summary of the results and a short outlook will be found in the Conclusion.

\section{The eigenvalue problem of a Hamiltonian with reflection}
\par
In this section, we give an analysis on the eigenvalue problem related to the Hamiltonian 
\begin{equation}
\label{H}
H=-\partial^2_x+v^2(x)-v'(x)R
\end{equation} 
where $v(x)$ is odd, $v(-x)=-v(x)$. 
First, let us notice that $H$ can be rewritten into the following two Hamiltonians 
\begin{equation}
\label{H_12}
  \begin{split}
  &H_1=-\partial^2_x+v^2(x)-v'(x):=-\partial^2_x+V_1(x), \\
  &H_2=-\partial^2_x+v^2(x)+v'(x):=-\partial^2_x+V_2(x),
  \end{split}
\end{equation}
by restricting $H$ on even or odd functions, respectively. 
In this way, we can avoid coping with the reflection $R$ directly. 
The potentials $V_1(x)$ and $V_2(x)$ are {\bf even} under the assumption that $v(x)$ is odd, 
and the relation $V_2(x)=V_1(x)+2v'(x)$ holds.
In this case the function $v(x)$ plays the role of the superpotential, 
the potentials $V_1(x)$ and $V_2(x)$ are known as a pair of supersymmetric partner potentials \cite{SUSYQM}. 

\par
In view of the property (\ref{H_eq}), let us 
assume that $H$ has a discrete sequence of eigenfunctions: 
$$
H\psi_{n}(x)=E_{n}\psi_{n}(x), \quad n=0,\pm 1, \pm 2, \ldots
$$
where $E_0=0$, $E_{-n}=E_{n}$, $\psi_{-n}(x)=\psi_{n}(-x)$, $n=1,2,\ldots$.
If we apply the decompositions 
\begin{equation}
\label{psi_e+o}
\psi_{\pm n}(x)=e_{n}(x)\pm o_{n}(x), 
\end{equation} 
where $e_{n}(x)$ and $o_{n}(x)$ are the even and odd components of $\psi_{\pm n}(x)$, respectively, 
then the eigenvalue equation of $H$ can be rewritten into 
$$
H\psi_{\pm n}(x)=H_{1}e_{n}(x)\pm H_{2}o_{n}(x)=E_{n}(e_{n}(x)\pm o_{n}(x))
$$
which lead to the following eigenvalue equations 
\begin{equation}
\label{H_12_EIeq}
H_1e_n(x)=E_ne_n(x),   \quad
H_2o_n(x)=E_no_n(x).
\end{equation}

\par

\begin{lemma}
\label{lemma_psi}
If there exist a sequence of eigenfunctions of $H_1$ which are all even and 
a sequence of eigenfunctions of $H_2$ which are all odd, 
and their eigenvalues satisfy the condition (\ref{EV_eo}), 
then the eigenfunctions of $H$ can be expressed as linear combinations of those of $H_1$ and $H_2$. 
\end{lemma}
\begin{proof}
Denote the eigenfunctions and eigenvalues of $H_1$ and $H_2$ by 
$\psi^{(1)}_{n}(x)$, $E^{(1)}_{n}$ and $\psi^{(2)}_{n}(x)$, $E^{(2)}_{n}$, respectively. 
Let $\psi^{(1)}_{N(n)}(x)$ be the even eigenfunctions of $H_1$ and 
$\psi^{(2)}_{M(n)}(x)$ be the odd eigenfunctions of $H_2$, 
where the indices $N(n)$ and $M(n)$ are both increasing, namely, 
$$
0\leq N(0)<N(1)<\cdots; \quad
0\leq M(0)<M(1)<\cdots.
$$
Then, by defining 
$$
e_{n}(x)=C^{(e)}\psi^{(1)}_{N(n)}(x), \quad
o_{n}(x)=C^{(o)}\psi^{(2)}_{M(n)}(x)
$$
with arbitrary constants $C^{(e)}$, $C^{(o)}$ which are not identically zero, 
and imposing the condition 
\begin{equation}
\label{EV_eo}
E^{(1)}_{N(n)}=E^{(2)}_{M(n)}
\end{equation}
one can solve the eigenvalue problem of $H$ as follow: 
\begin{equation}
\label{EVP_H_sol}
\psi_{n}(x)=C^{(e)}\psi^{(1)}_{N(n)}(x)+C^{(o)}\psi^{(2)}_{M(n)}(x), 
\quad E_{n}=E^{(1)}_{N(n)}=E^{(2)}_{M(n)}.
\end{equation}
\end{proof}

\begin{remark}
\label{rmk_SL}
\par
Recall that the eigenvalue problem of the operator $H_1$ ($H_2$) on some interval ($a,b$) with boundary conditions 
is a Sturm-Liouville problem. 
For example, the eigenvalue problem 
\begin{equation*}
\label{SHReq}
  \begin{split}
	\mathcal{H}\phi(x)=(-\partial^2_x+V(x))\phi(x)&=\lambda\phi(x), \quad x\in (a,b) \\
	c_{a}\phi(a)+d_{a}\phi'(a)&=0, \\
	c_{b}\phi(b)+d_{b}\phi'(b)&=0
  \end{split}
\end{equation*}
is a regular Sturm-Liouville problem. 
According to the Sturm-Liouville theory, 
if $V(x)$ is continuous and regular in ($a,b$) ,
then 
\begin{itemize}
\item[(1)] Eigenvalues of the operator $\mathcal{H}$ are real, simple and non-degenerate 
		(eigenvalues of different eigenfunctions are distinct, $\lambda_m\neq \lambda_n$, $\forall m\neq n$). 
		Further, the eigenvalues form an infinite sequence, and can be ordered according to increasing 
		magnitude so that $\lambda_0<\lambda_1<\cdots$ and $\lim_{n\rightarrow\infty}\lambda_n=\infty$.
\item[(2)] The eigenfunctions of $\mathcal{H}$ are orthogonal: 
$$
\int^{b}_{a}\phi_{m}(x)\phi_{n}(x)dx=0, \quad m\neq n.
$$
\end{itemize}
\end{remark}

\begin{remark}
\label{rmk_R}
\par
Since the reflection $R$ is involved in $H$, 
the range of the coordinate $x$ must be invariant under $R$, namely, it should be ($-a, a$) 
($a$ can be finite or infinite). 
Thus, the range of the coordinate $x$ in $H_1$ and $H_2$ should both be ($-a, a$). 
From the construction of $H_1$ and $H_2$ it is easily seen that once $H_1$ satisfies this condition, 
then the same holds for $H_2$. 
We shall distinguish the following two cases: 
(A) a ``genuine'' ($-a,a$) model and (B) two copies of ($0,a$) model. 
Note that the potential $V_1(x)$ ($V_2(x)$) is even here.

\par
In case (A), let us consider the operator $H_1$ on ($-a,a$). 
Assume that 
$H_1$ has an infinite sequence of eigenfunctions, 
$H_1\psi^{(1)}_{n}(x)=E^{(1)}_{n}\psi^{(1)}_{n}(x)$ ($n=0,1,\ldots$) 
with $0=E_0<E_1<\cdots$, 
then the eigenfunctions and eigenvalues of $H_2$ follow from the relations (\ref{rela_1}) and (\ref{rela_2}) automatically. 
Thus those of $H$ can be derived using Lemma \ref{lemma_psi}. 

\par
In case (B), let us consider the operator $H_1$ on ($0,a$). 
In this case, we assume that the potentials $V_1(x)$, $v(x)$ are singular at $x=0$, 
and $H_1$ has an infinite sequence of eigenfunctions, 
$H_1\psi^{(1)}_{n}(x)=E^{(1)}_{n}\psi^{(1)}_{n}(x)$ ($n=0,1,2,\ldots$) with $0=E_0<E_1<\cdots$. 
In particular, we assume that $\psi^{(1)}_{n}(x)$ satisfies $\psi^{(1)}_{n}(0)=(\psi^{(1)}_{n})'(0)=0$. 
Define the Hamiltonian $\tilde{H}_1$ of a new model on ($-a,a$) by $\tilde{H}_1=H_1$, 
which is singular at $x=0$. 
By defining the operators $\tilde{A}$ and $\tilde{A}^{\dagger}$ of new model on ($-a, a$) as 
$\tilde{A}=A$ and $\tilde{A}^{\dagger}=A^{\dagger}$, we have $\tilde{H}=\tilde{A}^{\dagger}\tilde{A}$. 
Let us define $\tilde{\psi}^{(1)}_n(x)$ on ($-a,a$) as follows: 
\begin{equation}
\label{tilde_psi}
\tilde{\psi}^{(1)}_{2n}(x)=
  \begin{cases}
	\psi^{(1)}_{n}(x) & (0\leq x<a) \\
	\psi^{(1)}_{n}(-x) & (-a<x<0)
  \end{cases}, \quad
\tilde{\psi}^{(1)}_{2n+1}(x)=
  \begin{cases}
	\psi^{(1)}_{n}(x) & (0\leq x<a) \\
	-\psi^{(1)}_{n}(-x) & (-a<x<0)
  \end{cases}
\end{equation}
which satisfy 
\begin{equation}
\label{tilde_psi_1}
\tilde{\psi}^{(1)}_n(-x)=(-1)^{n}\tilde{\psi}^{(1)}_n(x).
\end{equation}
Then we have 
$$
\tilde{H}_1\tilde{\psi}^{(1)}_n(x)=\tilde{E}_{n}\tilde{\psi}^{(1)}_n(x), \quad 
\tilde{E}_{2n}=\tilde{E}_{2n+1}=E_{n} \quad (n\geq 0). 
$$
The energy eigenvalues are doubly degenerate, which is allowed by the singularity of $x=0$. 
The eigenfunctions $\tilde{\psi}^{(1)}_{2n}(x)$ and $\tilde{\psi}^{(1)}_{2m+1}(x)$ are orthogonal: 
$\int^{a}_{-a}\tilde{\psi}^{(1)}_{2n}(x)\tilde{\psi}^{(1)}_{2m+1}(x)dx=0$. 
Again, using the relations (\ref{rela_1}), (\ref{rela_2}) and Lemma \ref{lemma_psi} one can obtain 
the eigenfunctions and eigenvaules of $H$. 
\end{remark}

\section{Supersymmetric quantum mechanics and shape invariant even potentials}
\par
In the theory of SUSY QM \cite{SUSYQM}, 
the Hamiltonians $H_1$ and $H_2$ can be factorized as follow: 
$$
H_1=A^{\dagger}A+E^{(1)}_{0}, \quad
H_2=AA^{\dagger}+E^{(1)}_{0},
$$
where $A^{\dagger}$ is the conjugation of $A$, 
$$
A=\partial_x+v(x), \quad A^{\dagger}=-\partial_x+v(x). 
$$

\par
If we consider the unbroken SUSY where the ground state energy is zero, namely, $E^{(1)}_{0}=0$, 
then we can choose $v(x)=-(\ln\psi^{(1)}_{0}(x))'$ such that $H\psi^{(1)}_{0}(x)=A^{\dagger}A\psi^{(1)}_{0}(x)=0$.
This follows from the fact that $A$ annihilates the ground state wave function $\psi^{(1)}_{0}(x)$. 

\par
Again, let us denote the eigenfunctions and eigenvalues of $H_1$ and $H_2$ by 
$\psi^{(1)}_{n}(x)$, $E^{(1)}_{n}$ and $\psi^{(2)}_{n}(x)$, $E^{(2)}_{n}$, respectively. 
They satisfy the equations 
$$
H_1\psi^{(1)}_{n}(x)=A^{\dagger}A\psi^{(1)}_{n}(x)=E^{(1)}_{n}\psi^{(1)}_{n}(x), \quad 
H_2A\psi^{(1)}_{n}(x)=AA^{\dagger}A\psi^{(1)}_{n}(x)=E^{(1)}_{n}A\psi^{(1)}_{n}(x),
$$
$$\hspace{4mm}
H_2\psi^{(2)}_{n}(x)=AA^{\dagger}\psi^{(2)}_{n}(x)=E^{(2)}_{n}\psi^{(2)}_{n}(x), \quad 
H_1A^{\dagger}\psi^{(2)}_{n}(x)=A^{\dagger}AA^{\dagger}\psi^{(2)}_{n}(x)=E^{(2)}_{n}A^{\dagger}\psi^{(2)}_{n}(x),
$$
from which it follows that 
\begin{equation}
\label{rela_1}
E^{(2)}_{n}=E^{(1)}_{n+1}, \quad E^{(1)}_{0}=0, 
\end{equation}
\begin{equation}
\label{rela_2}
\psi^{(2)}_{n}(x)=(C^{(1)}_{n})^{-1} A\psi^{(1)}_{n+1}(x), 
\end{equation}
\begin{equation}
\label{rela_3}
\psi^{(1)}_{n+1}(x)=(C^{(2)}_{n+1})^{-1} A^{\dagger}\psi^{(2)}_{n}(x),
\end{equation}
where the coefficients $C^{(1)}_{n}$ and  $C^{(2)}_{n}$ satisfy the condition $C^{(1)}_{n}C^{(2)}_{n+1}=E^{(1)}_{n+1}$. 
Note that the relations (\ref{rela_2}), (\ref{rela_3}) and the definitions of $A, A^{\dagger}$ imply that 
$\psi^{(1)}_{n+1}(x)$ and $\psi^{(2)}_{n}(x)$ have different parity, namely, 
if $\psi^{(1)}_{n+1}(x)$ is even, then $\psi^{(2)}_{n}(x)$ is odd; 
if $\psi^{(1)}_{n+1}(x)$ is odd, then $\psi^{(2)}_{n}(x)$ is even.

\par
In what follows we adopt the assumption in the proof of Lemma \ref{lemma_psi} 
that $\psi^{(1)}_{N(n)}(x)$ are even and $\psi^{(2)}_{M(n)}(x)$ are odd for $n=0,1,\ldots$. 
This can be achieved through some minor trick. 
(In fact, if $\psi^{(1)}_{n}(x)$ and $\psi^{(2)}_{n}(x)$ do not have definite parity, 
we can redefine them by 
$$
\tilde{\psi}^{(1)}_{n}(x):=\frac{1}{2}(\psi^{(1)}_{n}(x)+\psi^{(1)}_{n}(-x)), \quad
\tilde{\psi}^{(2)}_{n}(x):=\frac{1}{2}(\psi^{(2)}_{n}(x)-\psi^{(2)}_{n}(-x))
$$
which are still eigenfunctions of $H_1$ and $H_2$ with eigenvalues $E^{(1)}_{n}$ and $E^{(2)}_{n}$, respectively. 
This is because $\psi^{(1)}_{n}(-x)$ and $\psi^{(2)}_{n}(-x)$ are also eigenfunctions of 
$H_1$ and $H_2$ with eigenvalues $E^{(1)}_{n}$ and $E^{(2)}_{n}$, respectively. 
Obviously, the new defined eigenfunctions are even and odd, respectively.)

\par
If we let $v(x)$ in the operator $\mathcal{L}$ be defined by 
$v(x)=-(\ln\psi^{(1)}_{0}(x))'$, 
then the operators $A$ and $A^{\dagger}$ turn out to be the restrictions of $\mathcal{L}$ 
on even and odd functions, 
respectively. 
It follows that 
\begin{equation}
\label{L_psi_1}
\mathcal{L}\psi^{(1)}_{N(n)}(x)=A\psi^{(1)}_{N(n)}(x)=C^{(1)}_{N(n)-1}\psi^{(2)}_{N(n)-1}(x), 
\end{equation}
\begin{equation}
\label{L_psi_2}
\mathcal{L}\psi^{(2)}_{M(n)}(x)=A^{\dagger}\psi^{(2)}_{M(n)}(x)=C^{(2)}_{M(n)+1}\psi^{(1)}_{M(n)+1}(x). 
\end{equation}
From these relations we can derive the eigenfunctions of $\mathcal{L}$.

\begin{lemma}
\label{lemma_EV_L}
Let $C^{(1)}_{N(n)-1}=C^{(2)}_{N(n)}=\sqrt{E^{(1)}_{N(n)}}$, $n=0,1,\ldots$, and 
\begin{equation}
\label{psi_2}
\psi^{(2)}_{N(n)-1}(x)
=\left(\sqrt{E^{(1)}_{N(n)}}\right)^{-1}A\psi^{(1)}_{N(n)}(x).
\end{equation}
If we further assume that $M(n)=N(n)-1$, $n=0,1,\ldots$, 
then the eigenvalue problem $\mathcal{L}\psi_{\pm n}(x)=\lambda_{\pm n}\psi_{\pm n}(x)$ can be solved as follow: 
\begin{equation}
\label{opL_sol}
\psi_{\pm n}(x)=\psi^{(1)}_{N(n)}(x)\pm \psi^{(2)}_{N(n)-1}(x), \quad 
\lambda_{\pm n}=\pm\sqrt{E^{(1)}_{N(n)}}. 
\end{equation}
\end{lemma}
\begin{proof}
The condition $M(n)=N(n)-1$ is equivalent with (\ref{EV_eo}) in view of the relation (\ref{rela_1}), 
then it follows from Lemma \ref{lemma_psi} that the eigenfunctions of $H$ as well as those of $\mathcal{L}$ 
can be expressed in terms of the linear combination of $\psi^{(1)}_{n}(x)$ and $\psi^{(2)}_{n}(x)$. 
Using this condition and the relations (\ref{L_psi_1}), (\ref{L_psi_2}) 
one can easily check the following eigenvalue equations 
\begin{equation*}
\label{opL_psi_12}
\mathcal{L}\left(\psi^{(1)}_{N(n)}(x)\pm \sqrt{C^{(1)}_{N(n)-1}/C^{(2)}_{N(n)}}\psi^{(2)}_{N(n)-1}(x)\right)=
\pm\sqrt{E^{(1)}_{N(n)}}\left(\psi^{(1)}_{N(n)}(x)\pm \sqrt{C^{(1)}_{N(n)-1}/C^{(2)}_{N(n)}}\psi^{(2)}_{N(n)-1}(x)\right).
\end{equation*}
Then (\ref{opL_sol}) follows immediately from the assumption $C^{(1)}_{N(n)-1}=C^{(2)}_{N(n)}=\sqrt{E^{(1)}_{N(n)}}$. 
\end{proof}

\subsection{Shape invariant even potentials}
\par
It is now clear that once the even potentials $V_1(x), V_2(x)$ and 
their corresponding eigenfunctions and eigenvalues are known, 
then the eigenfunctions and eigenvalues of $\mathcal{L}$ with $v(x)$ given by 
the superpotential related to $V_1(x), V_2(x)$ 
follow automatically from Lemma \ref{lemma_EV_L}.

\par
A good class of potentials are the shape invariant potentials which satisfy the condition
\begin{equation}
\label{cond_SIP}
V_2(x;a_1)=V_1(x;a_2)+R(a_1),
\end{equation}
where $a_1$ is a set of parameters, $a_2$ is a translation of $a_1$, 
and it follows that 
$$
R(a_1)=V_2(x;a_1)-V_1(x;a_2)=V_1(x;a_1)-V_1(x;a_2)+2v'(x;a_1).
$$ 

\par
Unless otherwise stated, for any function $f(x)$ appear later we default $f(x)$ stands for $f(x;a_1)$, 
in other words, $a_1$ and $a_2$ must appear simultaneously. 
The condition (\ref{cond_SIP}) implies that 
$$
H_{2}\psi^{(1)}_{m}(x;a_2)=[E^{(1)}_{m}(a_2)+R(a_1)]\psi^{(1)}_{m}(x;a_2).
$$
By comparing this with the eigenvalue equation 
$H_{2}\psi^{(2)}_{n}(x;a_1)=E^{(2)}_{n}(a_1)\psi^{(2)}_{n}(x;a_1)$ 
we can conclude that if $E^{(2)}_{n}(a_1)=E^{(1)}_{m}(a_2)+R(a_1)$ holds for some indices $m$ and $n$, 
then 
\begin{equation}
\label{psi_12_0}
\psi^{(2)}_{n}(x;a_1)\propto \psi^{(1)}_{m}(x;a_2). 
\end{equation}

In particular, if $m=n$, then (\ref{psi_12_0}) becomes (\ref{psi_12}) and it means that 
the eigenfunctions of $H_{1}$ and $H_{2}$ coincide through a translation on certain parameters. 
Fortunately, it turns out that this is true for all the examples we shall consider in this paper (see Remark \ref{even_EVP}). 
Combining the relations (\ref{psi_2}) and (\ref{psi_12}), then we have 
\begin{equation}
\label{psi_1_a2}
\psi^{(2)}_{N(n)-1}(x;a_1)=\left(\sqrt{E^{(1)}_{N(n)}(a_1)}\right)^{-1}A\psi^{(1)}_{N(n)}(x;a_1)
\propto\psi^{(1)}_{N(n)-1}(x;a_2).
\end{equation}
Recall that $\psi^{(2)}_{N(n)-1}(x;a_1)$ is odd and $\psi^{(1)}_{N(n)}(x;a_1)$ is even, 
thus $\psi^{(1)}_{N(n)-1}(x;a_1)$ is odd too since the translation on the parameter(s) $a_1$ does not change the parity in $x$. 
So $\psi^{(1)}_{n}(x)$ is symmetric: 
\begin{equation}
\label{psi_1_sym}
\psi^{(1)}_{n}(-x)=(-1)^{n}\psi^{(1)}_{n}(x). 
\end{equation}
And $\psi^{(2)}_{n}(x)$ should also be symmetric due to the relation (\ref{rela_2}). 

\par
A list of shape invariant potentials and of the corresponding wavefunctions 
related with supersymmetric quantum mechanics is presented in \cite{SUSYQM} (Table 4.1).
We can readily obtain from this table the even potentials 
by putting restrictions on certain parameters. 
The results are given in Table 1 (which is split in two parts).
Specifically, the examples of 
{\bf shifted oscillator}, {\bf Scarf II or Rosen-Morse II (hyperbolic)} and {\bf Scarf I} potentials
belong to case (A) in Remark \ref{rmk_R}
while the examples of {\bf 3d oscillator}, {\bf generalized P\"oschl-Teller} and {\bf P\"oschl-Teller} potentials 
belong to case (B). 
For convenience we shall call them the type (A) examples and the type (B) examples, respectively. 

\begin{remark}
\label{even_EVP}
From Table 1 one can see that in all the examples the relation 
$$
E_{n}^{(1)}(a_2)+R(a_1)=E^{(2)}_n(a_1),
$$ 
holds, which leads to 
\begin{equation}
\label{psi_12}
\psi^{(2)}_{n}(x;a_1)\propto\psi^{(1)}_{n}(x;a_2).
\end{equation}
The above relation together with Lemma \ref{lemma_EV_L} implies that 
the eigenfunctions of $\mathcal{L}$ with $v(x)$ given by 
the superpotential in Table 1 can be written as:
\begin{equation} 
\label{psi_a_1}
\psi_{\pm n}(x;a_1)=\psi^{(1)}_{N(n)}(x;a_1)\pm \psi^{(2)}_{N(n)-1}(x;a_1)
=\psi^{(1)}_{N(n)}(x;a_1)\pm \tilde{C}_{n}\psi^{(1)}_{N(n)-1}(x;a_2), 
\end{equation}
where it follows from (\ref{psi_1_a2}) that 
\begin{equation} 
\label{psi_a_2_C}
\tilde{C}_{n}=\left(\sqrt{E^{(1)}_{N(n)}(a_1)}\right)^{-1}
\frac{A\psi^{(1)}_{N(n)}(x;a_1)}{\psi^{(1)}_{N(n)-1}(x;a_2)}.
\end{equation}

\begin{table}[h]
\centering
\caption{Shape invariant even potentials derived from \cite{SUSYQM} (Table 4.1), 
where the parameters $a_1$ and $a_2$ are related by a translation $a_2=a_1+\alpha$. 
Here we replaced the parameter $\omega$ in \cite{SUSYQM} with $\omega=2s^2$ 
for our convenience. Unless specified explicitly otherwise, the parameters $A,B,\alpha,s,l$ are all taken $\geq 0$,  
and the range of potentials is $-\infty\leq x\leq\infty$, $0\leq r\leq\infty$.}
\label{even potentials}
\begin{adjustbox}{width=\textwidth}
\begin{tabulary}{1.2\textwidth}{LCCCC} \toprule
Name of potential 		      & $v(x)$	      		  & $V_1(x;a_1)$ 				        & $y$ 	            	 & $\psi^{(1)}_n(y)$ \\ \midrule
shifted oscillator 		      & $s^{2}x$ 		      	  & $s^{2}(s^{2}x^2-1)$ 					& $sx$ 	                  & $e^{-\frac{1}{2}y^2}H_{n}(y)$  \\[4pt] \hline 
Scarf II or   	      	        &					  &								&				      &   \\
Rosen-Morse II			& $A\tanh(\alpha x)$ 	  & $A^2-A(A+\alpha)$sech$^{2}(\alpha x)$ & $\sinh(\alpha x)$ 	 & $i^{n}(y^2+1)^{-\frac{A}{2\alpha}}$  \\
(hyperbolic)			&					  &								&		& \hspace{-3mm}$\cdot P_{n}^{(-\frac{A}{\alpha}-\frac{1}{2},-\frac{A}{\alpha}-\frac{1}{2})}(iy)$  \\[4pt] \hline
Scarf I 				& $A\tan(\alpha x)$ 		  & $-A^2+A(A-\alpha)\sec^{2}(\alpha x)$ 	& $\sin(\alpha x)$	 & $(1-y^2)^{\frac{A}{2\alpha}}$  \\
(trigonometric)			& ($-\frac{\pi}{2\alpha}\leq x\leq \frac{\pi}{2\alpha}$)				  &								&		   & \hspace{1mm}$\cdot P_{n}^{(\frac{A}{\alpha}-\frac{1}{2},\frac{A}{\alpha}-\frac{1}{2})}(y)$  \\[4pt] \hline
3d oscillator & $s^{2}r-\frac{l+1}{r}$   & $s^{4}r^{2}+\frac{l(l+1)}{r^2}-(2l+3)s^2$ 	& $s^{2}r^{2}$    		 & $y^{\frac{l+1}{2}}e^{-\frac{y}{2}}L_{n}^{l+\frac{1}{2}}(y)$   \\[4pt] \hline
generalized   			& \hspace{-4.4mm}$A\coth(\alpha r)-B$cosech$(\alpha r)$ & $A^2+(B^2+A^2+A\alpha)$cosech$^2(\alpha r)$ & $\cosh(\alpha r)$   &  $(y-1)^{\frac{B-A}{2\alpha}}(y+1)^{-\frac{A+B}{2\alpha}}$ \\
P\"oschl-Teller 			& $(A<B)$			 	  & $-B(2A+\alpha)\coth(\alpha r)$cosech$(\alpha r)$	     &		      & \hspace{1mm}$\cdot P_{n}^{(\frac{B-A}{\alpha}-\frac{1}{2},-\frac{A+B}{\alpha}-\frac{1}{2})}(y)$  \\[4pt] \hline
P\"oschl-Teller			& \hspace{-4.4mm}$A\tan(\alpha r)-B\cot(\alpha r)$		& $-(A+B)^2+A(A-\alpha)\sec^2(\alpha r)$		& $\cos(2\alpha r)$		& $(1-y)^{\frac{B}{2\alpha}}(1+y)^{\frac{A}{2\alpha}}$ \\		
& \hspace{-5mm}($A,B>0$, $0<r<\frac{\pi}{2\alpha}$)		& $+B(B-\alpha)\text{cosec}^2(\alpha r)$	&	& $\cdot P^{(\frac{B}{\alpha}-\frac{1}{2},\frac{A}{\alpha}-\frac{1}{2})}(y)$  \\[4pt] \bottomrule
\end{tabulary}
\end{adjustbox}
\begin{adjustbox}{width=\textwidth}
\begin{tabulary}{1.2\textwidth}{LCCCCC}  \toprule
Name of potential 		      & $a_1$ 	      & $a_2$       & $E^{(1)}_n(a_1)$          & $E^{(2)}_n(a_1)$      & $E_{m}^{(1)}(a_2)+R(a_1)$    \\ \midrule
shifted \\ oscillator 		      & $s$        & $s$    & $2ns^{2}$      	        & $2(n+1)s^{2}$    	   & $2(m+1)s^2$      \\[4pt] \hline 
Scarf II or   	      	              &		    	      &			  &				        &				   &   \\
Rosen-Morse II			      & $A$	    	      & $A-\alpha$ & $2nA\alpha-n^2\alpha^2$   & $2(n+1)A\alpha-(n+1)^2\alpha^2$  & $2(m+1)A\alpha-(m+1)^2\alpha^2$ \\
(hyperbolic)			      &			      &			  &		     			&	&  \\[4pt] \hline
Scarf I \\(trigonometric)	     & $A$	 	      & $A+\alpha$  & $2nA\alpha+n^2\alpha^2$ & $2(n+1)A\alpha+(n+1)^2\alpha^2$	& $2(m+1)A\alpha+(m+1)^2\alpha^2$  \\[4pt] \hline
3d oscillator & $l$ 	      	      & $l+1$	      & $4ns^{2}$  & $4(n+1)s^{2}$ 	        & $4(m+1)s^2$   \\[4pt] \hline
generalized \\ P\"oschl-Teller  & $A$  	      & $A-\alpha$ & $2nA\alpha-n^2\alpha^2$    &  $2(n+1)A\alpha-(n+1)^2\alpha^2$      & $2(m+1)A\alpha-(m+1)^2\alpha^2$      \\[4pt] \hline
P\"oschl-Teller 			     & 	$A, B$	  	      &	 $A+\alpha, B+\alpha$	   & $4n\alpha(A+B+n\alpha)$	& 
$4(n+1)\alpha(A+B+(n+1)\alpha)$       & $4(m+1)\alpha(A+B+(m+1)\alpha)$     \\[4pt] \bottomrule
\end{tabulary}
\end{adjustbox}
\end{table}


\par
Recall that the indices $N(n)$, $n=0,1,\ldots$, are chosen in such a way that 
$\psi^{(1)}_{N(n)}(x)$ are even and $\psi^{(1)}_{N(n)-1}(x)$ are odd. 
Since the Hermite polynomials $H_{n}(x)$ are symmetric, 
and the Jacobi polynomials $P^{(\alpha,\beta)}_{n}(x)$ are symmetric when $\alpha=\beta$, 
we observe in Table 1 that all the eigenfunctions of the type (A) examples (the first three examples) are symmetric, 
and those of the type (B) examples can be constructed as symmetric functions using the method (\ref{tilde_psi}) 
introduced in Remark \ref{rmk_R}. 
Then it turns out that $N(n)=2n$, $n=0,1,\ldots$, 
and (\ref{psi_a_1}) reads 
\begin{equation}
\label{psi_pm_ex}
\psi_{\pm n}(x;a_1)=\psi^{(1)}_{2n}(x;a_1)\pm \tilde{C}_{n}\psi^{(1)}_{2n-1}(x;a_2). 
\end{equation}
\end{remark}

\par
To summarize, we provide a list of these eigenfunctions in Table \ref{eigenfunc_operaotL}. 
For the explicit definitions and properties of these classical orthogonal polynomials 
(Hermite, Laguerre and Jacobi polynomials) one can refer to the Appendix. 
\begin{table}[htb]
\centering
\caption{Eigenfunctions of the operator $\mathcal{L}=\partial_x R+v(x)$ with $v(x)$ given by the superpotentials in table 1.}
\label{eigenfunc_operaotL}
\begin{tabular}{ccl} \hline
$v(x)$   			& $\tilde{C}_{n}$	&  $\psi_{\pm n}(x)$  \\[2pt] \hline
$s^2x$			& $2\sqrt{n}$	& $e^{-\frac{1}{2}s^2x^2}[H_{2n}(sx)\pm\tilde{C}_{n}H_{2n-1}(sx)]$ \\[2pt] \hline
$A\tanh(\alpha x)$	& \multirow{2}{*}{$\frac{1}{2}\sqrt{\frac{A-n\alpha}{n\alpha}}$}	& $(-1)^n\cosh(\alpha x)^{-\frac{A}{\alpha}}[P_{2n}^{(-\frac{A}{\alpha}-\frac{1}{2}),-\frac{A}{\alpha}-\frac{1}{2})}
				(i\sinh(\alpha x))$ \\ 
  				&  & \hspace{2.6cm}$\mp\tilde{C}_{n}i\cosh(\alpha x)P_{2n-1}^{(-\frac{A}{\alpha}+\frac{1}{2}),-\frac{A}{\alpha}+
				\frac{1}{2})}
				(i\sinh(\alpha x))]$  \\[2pt] \hline
$A\tan(\alpha x)$	& $\frac{1}{2}\sqrt{\frac{A+n\alpha}{n\alpha}}$	& $|\cos(\alpha x)|^{\frac{A}{\alpha}}[P_{2n}^{(\frac{A}{\alpha}-\frac{1}{2},\frac{A}{\alpha}-\frac{1}{2})}(\sin(\alpha x))
				\pm\tilde{C}_{n}|\cos(\alpha x)|P_{2n-1}^{(\frac{A}{\alpha}+\frac{1}{2},\frac{A}{\alpha}+\frac{1}{2})}(\sin(\alpha x)]$  
				\\[2pt] \hline
$s^2r-\frac{l+1}{r}$    & $-\frac{1}{\sqrt{n}}$	& $|sr|^{l+1}e^{-\frac{1}{2}s^2r^2}[L_n^{(l+\frac{1}{2})}(s^2r^2)\pm\tilde{C}_{n}srL_{n-1}^{(l+\frac{3}{2})}(s^2r^2)]$ \\[2pt] \hline
$A\coth(\alpha r)$	& \multirow{2}{*}{$\frac{1}{2}\sqrt{\frac{2A-n\alpha}{n\alpha}}$}	& $(\cosh(\alpha r)-1)^{\frac{B-A}{2\alpha}}(\cosh(\alpha r)+1)^{-\frac{B+A}
				{2\alpha}}[P_{n}^{(\frac{B-A}{\alpha}-\frac{1}{2},-\frac{B+A}{\alpha}-\frac{1}{2})}(\cosh(\alpha r)$  \\
$-B\text{cosech}(\alpha r)$	& 	& \hspace{3cm}$\pm
\tilde{C}_{n}\sinh(\alpha r)P_{n-1}^{(\frac{B-A}{\alpha}+\frac{1}{2},-\frac{B+A}
				{\alpha}+\frac{1}{2})}(\cosh(\alpha r)]$  \\[2pt] \hline
$A\tan(\alpha r)$	& \multirow{2}{*}{$\frac{1}{2}\sqrt{\frac{A+B+n\alpha}{n\alpha}}$}	&  $(1-\cos(2\alpha r))^{\frac{B}{2\alpha}}(1+\cos(2\alpha r))^{\frac{A}{2\alpha}}
[P_{n}^{(\frac{B}{\alpha}-\frac{1}{2},\frac{A}{\alpha}-\frac{1}{2})}(\cos(2\alpha r))$  \\
$-B\cot(\alpha r)$	&	& \hspace{3cm}$\pm\tilde{C}_{n}\sin(2\alpha r)
P_{n}^{(\frac{B}{\alpha}+\frac{1}{2},\frac{A}{\alpha}+\frac{1}{2})}(\cos(2\alpha r))]$  \\[2pt] \hline
\end{tabular}
\end{table}

\section{Dunkl-SUSY orthogonal functions in terms of classical orthogonal polynomials}
\par
In this section we shall give some examples of Dunkl-SUSY orthogonal functions explicitly. 
Before that we may apply a gauge transformation on the operator $\mathcal{L}$ as follow: 
\begin{equation}
\label{Y}
Y:=(\psi^{(1)}_0(x))^{-1}\mathcal{L}\psi^{(1)}_0(x)=\partial_x R+v(x)(I-R).
\end{equation}
The eigenfunctions of the new operator $Y$ are 
$Q_n(x)=(\psi^{(1)}_0(x))^{-1}\psi^{(1)}_n(x)$ ($n=0,1,\ldots$). 
We will show that these eigenfunctions also satisfy certain orthogonality relations, 
that will deem giving them the name of 
Dunkl-Supersymmetric (Dunkl-SUSY) orthogonal functions.

\par
The weight function $\omega(x)$ associated with the operator $Y$ 
satisfies to \cite{SUSYQM, XBI}
$$
\frac{\omega'(x)}{\omega(x)}=-2v(x)
=2\frac{\psi'_0(x)}{\psi_0(x)}
$$
and hence $\omega(x)=(\psi^{(1)}_0(x))^2$. 
Therefore the orthogonality relation of $\{Q_n(x)\}_{n=0, \pm 1, \pm 2, \ldots}$ are 
\begin{equation}
\int_{I}(\psi^{(1)}_0(x))^2Q_n(x)Q_m(x)=0, \quad n\neq m,
\end{equation}
where the interval $I$ will be determined from the weight function $(\psi^{(1)}_0(x))^2$.

\par
With an eye to presenting a model-independent description of Dunkl-SUSY orthogonal functions, 
we now extract from Table \ref{eigenfunc_operaotL} the following families of such orthogonal functions 
that are defined in terms of classical orthogonal polynomials. 
We assume that all the Hermite, Laguerre, Jacobi polynomials 
($\hat{H}_n(x)$, $\hat{L}_n^{(\alpha)}(x)$, $\hat{P}_n^{(\alpha,\beta)}(x)$) 
involved are orthonormal: 
$$
\int^{\infty}_{-\infty}\hat{H}_{m}(x)\hat{H}_{n}(x)=\delta_{m,n}, \quad 
\int^{\infty}_{0}\hat{L}^{(\alpha)}_{m}(x)\hat{L}^{(\alpha)}_{n}(x)=\delta_{m,n}, \quad 
\int^{1}_{-1}\hat{P}^{(\alpha,\beta)}_{m}(x)\hat{P}^{(\alpha,\beta)}_{n}(x)=\delta_{m,n}.
$$
Specifically, let $H_n(x)$, $L_n^{(\alpha)}(x)$, $P_n^{(\alpha,\beta)}(x)$ be defined as in the Appendix, then
$$
\hat{H}_{n}(x)=(2^n n!\sqrt{\pi})^{-\frac{1}{2}}H_{n}(x), \quad
\hat{L}^{(\alpha)}_{n}(x)=\left(\frac{\Gamma(n+\alpha+1)}{n!}\right)^{-\frac{1}{2}}L^{(\alpha)}_{n}(x), 
$$
$$
\hat{P}^{(\alpha,\beta)}_{n}(x)=\left(\frac{2^{\alpha+\beta+1}\Gamma(\alpha+n+1)\Gamma(\beta+n+1)}
{n!\Gamma(\alpha+\beta+n+1)(\alpha+\beta+2n+1)}\right)^{-\frac{1}{2}}P^{(\alpha,\beta)}_{n}(x). 
$$

\par
Now we are ready to provide the following examples of Dunkl-SUSY orthogonal functions. 
\begin{itemize}
\item Dunkl-SUSY orthogonal functions in terms of the orthonormal Hermite polynomials $\hat{H}_n(x)$, 
which is related with the {\bf shifted oscillator} potential: 
\begin{eqnarray}
\label{SUSY OPs_QH}
Q^{(H)}_{\pm n}(x) \hspace{-2.4mm}&=&\hspace{-2.4mm}  
\frac{1}{\sqrt{2}}\bigg(\hat{H}_{2n}(sx)\pm \hat{H}_{2n-1}(sx)\bigg), 
\quad n\geq 1, \quad Q^{(H)}_{0}(x)=1, \\
Y \hspace{-2.4mm}&=&\hspace{-2.4mm} \partial_x R+s^2x, \quad 
YQ^{(H)}_{\pm n}(x)=\pm \sqrt{E_{2n}}Q^{(H)}_{\pm n}(x), \quad E_{n}=2ns^2, \\
& & \int^{\infty}_{-\infty}e^{-s^2x^2}Q^{(H)}_{n}(x)Q^{(H)}_{m}(x)=\delta_{nm}, 
\quad m,n\in\mathbb{Z}.
\end{eqnarray}

\item Dunkl-SUSY orthogonal functions in terms of the orthonormal Laguerre polynomials $\hat{L}^{(\alpha)}_{n}(x)$, 
which is related with the {\bf 3d oscillator} potential ($l+\frac{1}{2}\rightarrow\alpha$): 
\begin{eqnarray}
\label{SUSY OPs_QL1}
Q^{(L)}_{\pm n}(x) \hspace{-2.4mm}&=&\hspace{-2.4mm}  
\frac{1}{\sqrt{2}}\bigg(\hat{L}^{(\alpha)}_{n}(s^2x^2)\mp x\hat{L}^{(\alpha+1)}_{n-1}(s^2x^2)\bigg), 
\quad n\geq 1, \quad Q^{(L)}_{0}(x)=1, \\
Y \hspace{-2.4mm}&=&\hspace{-2.4mm} \partial_x R+s^2x-\frac{\alpha+1/2}{x}, \quad 
YQ^{(L)}_{\pm n}(x)=\pm\sqrt{E_n}Q^{(L)}_{\pm n}(x), \quad E_{n}=4ns^2,  \\
& & \int^{\infty}_{-\infty}e^{-s^2x^2}|sx|^{2\alpha+1}Q^{(L)}_{n}(x)Q^{(L)}_{m}(x)=\delta_{nm}, 
\quad m,n\in\mathbb{Z}.
\end{eqnarray}
(In fact, the above example returns to the example of the Hermite case when $\alpha=-\frac{1}{2}$, 
this is due to the relations \ref{H2L_1} and \ref{H2L_2}. 
It is also obvious from the eigenvalue equation of $Y$.)

\item Dunkl-SUSY orthogonal functions in terms of the orthonormal Jacobi polynomials $\hat{P}^{(\alpha,\beta)}_{n}(x)$: 

the example related with the {\bf Scarf II or Rosen-Morse II (hyperbolic)} potential, 
\begin{eqnarray}
\label{SUSY OPs_QJ1}
& & \hspace{-12mm} Q^{(J,1)}_{\pm n}(x) 
\frac{(-1)^n}{\sqrt{2}}\bigg(\hat{P}^{(-\frac{A}{\alpha}-\frac{1}{2},-\frac{A}{\alpha}-\frac{1}{2})}_{2n}(i\sinh(\alpha x))\pm
\cosh(x)\hat{P}^{(-\frac{A}{\alpha}+\frac{1}{2},-\frac{A}{\alpha}+\frac{1}{2})}_{2n-1}(i\sinh(\alpha x))\bigg),  \hspace{1mm} n\geq 1, \\ 
& & \hspace{-6mm} Y=\partial_x R+A\tanh(\alpha x), \quad 
YQ^{(J,1)}_{\pm n}(x)=\pm\sqrt{E_{2n}}Q^{(J,.1)}_{\pm n}(x), \quad E_{n}=2nA\alpha-n^2\alpha^2,  \\
& & \int^{\frac{i\pi}{2\alpha}}_{-\frac{i\pi}{2\alpha}}|\cosh(\alpha x)|^{-\frac{A}{\alpha}}Q^{(J,1)}_{n}(x)Q^{(J,1)}_{m}(x)=\delta_{nm}, 
\quad m,n\in\mathbb{Z};
\end{eqnarray}

the example related with the {\bf Scarf I (trigonometric)} potential,
\begin{eqnarray}
\label{SUSY OPs_QJ2}
& & \hspace{-12mm} Q^{(J,2)}_{\pm n}(x) =
\frac{1}{\sqrt{2}}\bigg(\hat{P}^{(\frac{A}{\alpha}-\frac{1}{2},\frac{A}{\alpha}-\frac{1}{2})}_{2n}(\sin(\alpha x)) \pm 
\cos(\alpha x)\hat{P}^{(\frac{A}{\alpha}+\frac{1}{2},\frac{A}{\alpha}+\frac{1}{2})}_{2n-1}(\sin(\alpha x))\bigg),  
\quad n\geq 1, \\
& & \hspace{-6mm}  Y=\partial_x R+A\tan(\alpha x), \quad 
YQ^{(J,2)}_{\pm n}(x)=\pm\sqrt{E_{2n}}Q^{(J,2)}_{\pm n}(x), \quad E_{n}=2nA\alpha+n^2\alpha^2,  \\
& & \int^{\frac{\pi}{2\alpha}}_{-\frac{\pi}{2\alpha}}|\cos(\alpha x)|^{-\frac{A}{\alpha}}Q^{(J,2)}_{n}(x)Q^{(J,2)}_{m}(x)=\delta_{nm}, 
\quad m,n\in\mathbb{Z};
\end{eqnarray}

the example related with the {\bf generalized P\"oschl-Teller} potential ($A<B$),
\begin{equation}
\label{SUSY OPs_QJ3}
Q^{(J,3)}_{\pm n}(x) =
\frac{1}{\sqrt{2}}\bigg(\hat{P}^{(\frac{B-A}{\alpha}-\frac{1}{2},-\frac{B+A}{\alpha}-\frac{1}{2})}_{n}(\cosh(\alpha x)) \mp 
i\sinh(\alpha x)\hat{P}^{(\frac{B-A}{\alpha}+\frac{1}{2},-\frac{B+A}{\alpha}+\frac{1}{2})}_{n-1}(\cosh(\alpha x))\bigg),  
\hspace{1mm} n\geq 1, 
\end{equation}
$$
Q^{(J,3)}_{\pm 0}(x) = 1,
$$
$$
Y=\partial_x R+A\coth(\alpha x)-B\text{cosech}(\alpha x), 
$$
\begin{equation}
YQ^{(J,3)}_{\pm n}(x)=\pm\sqrt{E_n}Q^{(J,3)}_{\pm n}(x), \quad E_{n}=2nA\alpha-n^2\alpha^2,  
\end{equation}
\begin{equation}
\int^{\frac{i\pi}{\alpha}}_{-\frac{i\pi}{\alpha}}(\cosh(\alpha r)-1)^{\frac{B-A}{2\alpha}}(\cosh(\alpha r)+1)^{-\frac{B+A}{2\alpha}}
Q^{(J,3)}_{n}(x)Q^{(J,3)}_{m}(x)=\delta_{nm}, 
\quad m,n\in\mathbb{Z};
\end{equation}

the example related with the {\bf P\"oschl-Teller} potential ($A,B>0$),
\begin{equation}
\label{SUSY OPs_QJ4}
Q^{(J,4)}_{\pm n}(x) =
\frac{1}{\sqrt{2}}\bigg(\hat{P}^{(\frac{B}{\alpha}-\frac{1}{2},\frac{A}{\alpha}-\frac{1}{2})}_{n}(\cos(2\alpha x)) \pm 
\sin(2\alpha x)\hat{P}^{(\frac{B}{\alpha}+\frac{1}{2},\frac{A}{\alpha}+\frac{1}{2})}_{n-1}(\cos(2\alpha x))\bigg),  
\quad n\geq 1, 
\end{equation}
$$
Q^{(J,4)}_{\pm 0}(x) = 1,
$$
$$
Y=\partial_x R+A\tan(\alpha x)-B\cot(\alpha x), 
$$
\begin{equation}
YQ^{(J,4)}_{\pm n}(x)=\pm\sqrt{E_n}Q^{(J,4)}_{\pm n}(x), \quad E_{n}=4n\alpha(A+B+n\alpha),  
\end{equation}
\begin{equation}
\int^{\frac{\pi}{2\alpha}}_{-\frac{\pi}{2\alpha}}(1-\cos(2\alpha x))^{\frac{B}{\alpha}}(1+\cos(2\alpha x))^{\frac{A}{\alpha}}
Q^{(J,4)}_{n}(x)Q^{(J,2)}_{m}(x)=\delta_{nm}, 
\quad m,n\in\mathbb{Z}.
\end{equation}

\end{itemize}

\section{The recurrence relation of the Dunkl-supersymmetric orthogonal polynomials}
\par 
Notice that in the previous section 
the examples of Dunkl-SUSY orthogonal functions in terms of 
the Hermite polynomials and the Laguerre polynomials 
are also polynomials, we call them Dunkl-supersymmetric orthogonal polynomials (Dunkl-SUSY OPs). 
Using these examples
we shall identify the main properties of these polynomials so as to offer in this section a characterization which is more intrinsic.
The most fundamental features of the Dunkl-SUSY OPs can be identified as:
\begin{enumerate}
\item[(A)] For all positive and negative integers $n$, the polynomial system $\{Q_n(x)\}_{n=0,\pm 1,\pm 2,\ldots}$ 
satisfy an orthogonality relation, 
$$
\int_{I}Q_{n}(x)Q_{m}(x)\omega(x)dx=h_n\delta_{n,m}, \quad 
(n,m=\ldots, -1,0,1,\ldots);
$$
\item[(B)] 
The polynomial $Q_{-n}(x)$ with negative index has the same degree as the polynomial $Q_{n}(x)$ with positive index,
$$
Q_{-n}(x):=R[Q_{n}(x)]=Q_{n}(-x). \quad (n=1, 2, \ldots);
$$
\item[(C)]  $\{Q_n(x)\}$ are the polynomial parts of the eigenfunctions of a Dunkl-type differential operator of the form 
$$\mathcal{L}=\partial_x R+v(x), \quad (v(-x)=-v(x)).$$
\end{enumerate}

\par
Let us now address the question of what can be said about the polynomial system $\{Q_n(x)\}_{n=0,\pm 1,\pm 2,\ldots}$ satisfying the conditions (A) and (B) if it is not assumed that they satisfy an eigenvalue equation.
The answer to this question is given by the following theorem. 
Without loss of generality,  from now on we take $Q_n(x)$ monic, 
i.e., with the coefficient of the highest degree term in $x$ equal to 1.

\begin{theorem}\label{thm5.1}
\par
A necessary and sufficient condition for the existence of a polynomial system 
$\{Q_n(x)\}_{n=0,\pm 1,\pm 2,\ldots}$ 
which satisfies the conditions (A) and (B) is that $Q_n(x)$ are expressed as 
\begin{equation}
\label{NSOP}
\begin{split}
Q_{n}(x) &= S_{2n}(x)+a_{n}S_{2n-1}(x), \\
Q_{-n}(x) &= S_{2n}(x)-a_{n}S_{2n-1}(x), 
\end{split}
\quad n=1,2,\ldots
\end{equation}
with $Q_{0}(x)=1$, where the coefficients $a_n$ depend on the polynomials $S_n(x)$ (see (\ref{Coe_a})), 
and with $\{S_n(x)\}_{n=0,1,2\ldots}$ a monic symmetric orthogonal polynomial system:
$$
S_n(-x) = (-1)^n S_n(x), \quad
\int_{I}S_{n}(x)S_{m}(x)\omega(x)dx=k_n\delta_{n,m}.
$$
\end{theorem}

\begin{proof}
The sufficiency is straightforward. If (\ref{NSOP}) holds, then it immediately follows that 
$Q_{-n}(x)=Q_{n}(-x)$, thus (B) is satisfied. 
For all nonnegative integers $n,m$, we have
\begin{equation*}
\begin{split}
\int_{I}Q_{n}(x)Q_{m}(x)\omega(x)dx &=
\int_{I}(S_{2n}(x)+a_{n}S_{2n-1}(x))(S_{2m}(x)+a_{m}S_{2m-1}(x))\omega(x)dx \\
&\hspace{-2cm}= 
(k_{2n}+a_{n}a_{m}k_{2n-1})\delta_{n,m}+a_{m}k_{2n}\delta_{2n,2m-1}+a_{n}k_{2n-1}\delta_{2n-1,2m} \\
&\hspace{-2cm}= 
(k_{2n}+a_{n}a_{m}k_{2n-1})\delta_{n,m},
\end{split}
\end{equation*}
\begin{equation*}
\begin{split}
\int_{I}Q_{n}(x)Q_{-m}(x)\omega(x)dx &=
\int_{I}(S_{2n}(x)+a_{n}S_{2n-1}(x))(S_{2m}(x)-a_{m}S_{2m-1}(x))\omega(x)dx \\
&\hspace{-2cm}= 
(k_{2n}-a_{n}a_{m}k_{2n-1})\delta_{n,m}-a_{m}k_{2n}\delta_{2n,2m-1}+a_{n}k_{2n-1}\delta_{2n-1,2m} \\
&\hspace{-2cm}= 
(k_{2n}-a_{n}a_{m}k_{2n-1})\delta_{n,m},
\end{split}
\end{equation*}
\begin{equation*}
\begin{split}
\int_{I}Q_{-n}(x)Q_{-m}(x)\omega(x)dx &=
\int_{I}(S_{2n}(x)-a_{n}S_{2n-1}(x))(S_{2m}(x)-a_{m}S_{2m-1}(x))\omega(x)dx \\
&\hspace{-2cm}= 
(k_{2n}+a_{n}a_{m}k_{2n-1})\delta_{n,m}-a_{m}k_{2n}\delta_{2n,2m-1}-a_{n}k_{2n-1}\delta_{2n-1,2m} \\
&\hspace{-2cm}= 
(k_{2n}+a_{n}a_{m}k_{2n-1})\delta_{n,m}.
\end{split}
\end{equation*}
Therefore, condition (A) is also satisfied. 
Besides, from the first and the third equation we also have 
\begin{equation}
\label{Const_h}
h_{n}=h_{-n}=k_{2n}+a_{n}^{2}k_{2n-1}.
\end{equation}

\par
As for the necessity, suppose that $\{Q_n(x)\}_{n=0,\pm 1,\pm 2,\ldots}$ satisfies the conditions (A) and (B). 
For $n=1,2,\ldots$, $Q_n(x)$ can be expressed as $Q_n(x)=e_n(x)+o_n(x)$, 
where $e_n(x)$ and $o_n(x)$ are even and odd polynomials, respectively. 
Then from condition (B) we have $Q_{-n}(x)=e_n(x)-o_n(x)$, while condition (A) implies that 
for any positive integers $n\neq m$, one has the relations 
\begin{equation}
\label{OrRe1}
0=\langle Q_n,Q_m \rangle
=\langle e_n,e_m \rangle+\langle e_n,o_m \rangle+\langle o_n,e_m \rangle+\langle o_n,o_m \rangle,
\end{equation}
\begin{equation}
\label{OrRe2}
0=\langle Q_n,Q_{-m} \rangle
=\langle e_n,e_m \rangle-\langle e_n,o_m \rangle+\langle o_n,e_m \rangle-\langle o_n,o_m \rangle,
\end{equation}
\begin{equation}
\label{OrRe3}
0=\langle Q_{-n},Q_m \rangle
=\langle e_n,e_m \rangle+\langle e_n,o_m \rangle-\langle o_n,e_m \rangle-\langle o_n,o_m \rangle,
\end{equation}
\begin{equation}
\label{OrRe4}
0=\langle Q_{-n},Q_{-m} \rangle
=\langle e_n,e_m \rangle-\langle e_n,o_m \rangle-\langle o_n,e_m \rangle+\langle o_n,o_m \rangle,
\end{equation}
which together lead to 
\begin{equation}
\label{OrRe_eo1}
\langle e_n,e_m \rangle=\langle e_n,o_m \rangle=\langle o_n,e_m \rangle=\langle o_n,o_m \rangle=0,
\end{equation}
where the inner product $\langle f,g \rangle=\int_{I}f(x)g(x)\omega(x)dx$. 
Note that (\ref{OrRe2}) and (\ref{OrRe3}) also hold for $n=m$, which implies that 
\begin{equation}
\label{OrRe_eo2}
\langle e_n,e_n \rangle=\langle o_n,o_n \rangle. \quad 
\end{equation}
The relations (\ref{OrRe_eo1}) and (\ref{OrRe_eo2}) mean that the polynomials $\{e_n(x),o_n(x)\}_{n=0,1,2,\ldots}$ 
form an orthogonal polynomial system, more exactly, in view of the parities of $e_n(x)$ and $o_n(x)$,  
they form a symmetric orthogonal polynomial system: 
$$
e_n(x)=S_{2n}(x), \quad o_n(x)=a_{n}S_{2n-1}(x),
$$
where the coefficients $a_n$ can be obtained from (\ref{OrRe_eo2}) and are: 
\begin{equation}
\label{Coe_a}
a_{n}=\sqrt{\frac{\langle S_{2n}(x),S_{2n}(x) \rangle}{\langle S_{2n-1}(x),S_{2n-1}(x) \rangle}}
=\sqrt{\frac{k_{2n}}{k_{2n-1}}}, \quad
(n=1,2,\ldots).
\end{equation}
Note that the subscripts in $S_{2n}(x)$ and $S_{2n-1}(x)$ do not necessarily represent the corresponding degrees. 
We have hence shown that
the conditions (A) and (B) lead to expression (\ref{NSOP}), 
thus proving necessity.
\end{proof}
Theorem \ref{thm5.1} provides a general presentation
of the Dunkl-SUSY OPs.
Conversely, if we are given a set of OPs satisfying the conditions (A) and (B), we can always recover the corresponding 
set of symmetric OPs $S_{n}(x)$. 

\par
Moreover, according to the relations (\ref{Const_h}) and (\ref{Coe_a}), 
the orthogonality constant of the polynomials defined by (\ref{NSOP}) turn out to be 
\begin{equation}
\label{hh}
h_{n}=h_{-n}=2k_{2n}, \quad (n=1,2,\ldots)
\end{equation}
and $h_{0}=k_{0}$, where $k_{n}$ are the orthogonality constant of $\{S_{n}(x)\}_{n=0,1,2,\ldots}$. 

\par
The recurrence relations can be given as follow.

\begin{theorem}
Let the monic symmetric OPs $\{S_{n}(x)\}_{n=0,1,2,\ldots}$ 
defined by the three-term recurrence relation: 
$$
S_{n}(x)=xS_{n-1}(x)-\gamma_{n}S_{n-2}(x), \quad (n=1,2,\ldots )
$$
with $S_{-1}(x)=0$, $S_{0}(x)=1$, 
then the monic polynomial system $\{Q_n(x)\}_{n=0,\pm 1,\pm 2,\ldots}$ defined by {\rm (\ref{NSOP})}
satisfies the recurrence relations: 
\begin{equation}
\label{RecuRe1}
\begin{split}
Q_{n+1}(x)=
\frac{1}{2}\bigg[x^2+(a_{n+1}-\frac{\gamma_{2n+1}}{a_{n}})x-\gamma_{2n+2}
-\frac{\gamma_{2n+1}a_{n+1}}{a_{n}}\bigg]Q_{n}(x) \\
+\frac{1}{2}\bigg[x^2+(a_{n+1}+\frac{\gamma_{2n+1}}{a_{n}})x-\gamma_{2n+2}
+\frac{\gamma_{2n+1}a_{n+1}}{a_{n}}\bigg]Q_{-n}(x),
\end{split}
\quad n=1,2,\ldots
\end{equation}
\begin{equation}
\label{RecuRe2}
\begin{split}
Q_{-(n+1)}(x)=
\frac{1}{2}\bigg[x^2-(a_{n+1}+\frac{\gamma_{2n+1}}{a_{n}})x-\gamma_{2n+2}
+\frac{\gamma_{2n+1}a_{n+1}}{a_{n}}\bigg]Q_{n}(x)  \\
+\frac{1}{2}\bigg[x^2-(a_{n+1}-\frac{\gamma_{2n+1}}{a_{n}})x-\gamma_{2n+2}
-\frac{\gamma_{2n+1}a_{n+1}}{a_{n}}\bigg]Q_{-n}(x).
\end{split}
\quad n=1,2,\ldots
\end{equation}
with $Q_{0}(x)=1$.
\end{theorem}

\begin{proof}
First, it follows from  (\ref{NSOP}) that 
\begin{equation}
\label{SS}
S_{2n}(x)=\frac{Q_{n}(x)+Q_{-n}(x)}{2}, \quad 
S_{2n-1}(x)=\frac{Q_{n}(x)-Q_{-n}(x)}{2a_{n}}.
\end{equation}
By definition, we have 
$$
Q_{n+1}(x)=S_{2n+2}(x)+a_{n+1}S_{2n+1}(x)
=(x+a_{n+1})S_{2n+1}(x)-\gamma_{2n+2}S_{2n}(x)
$$
$$\hspace{-6.5mm}
=(x+a_{n+1})(xS_{2n}(x)-\gamma_{2n+1}S_{2n-1})-\gamma_{2n+2}S_{2n}(x)
$$
$$\hspace{0mm}
=(x^2+a_{n+1}x-\gamma_{2n+2})S_{2n}(x)-\gamma_{2n+1}(x+a_{n+1})S_{2n-1},
$$
where the three-term recurrence relation of $\{S_{n}(x)\}$ has been used twice. 
Substituting (\ref{SS}) into the above then leads to 
$$
2Q_{n+1}(x)=
(x^2+a_{n+1}x-\gamma_{2n+2})(Q_{n}(x)+Q_{-n}(x))
-\frac{\gamma_{2n+1}}{a_{n}}(x+a_{n+1})(Q_{n}(x)-Q_{-n}(x))
$$
$$\hspace{-9.6mm}
=\bigg[x^2+(a_{n+1}-\frac{\gamma_{2n+1}}{a_{n}})x-\gamma_{2n+2}
-\frac{\gamma_{2n+1}a_{n+1}}{a_{n}}\bigg]Q_{n}(x)
$$
$$\hspace{-5.6mm}
+\bigg[x^2+(a_{n+1}+\frac{\gamma_{2n+1}}{a_{n}})x-\gamma_{2n+2}
+\frac{\gamma_{2n+1}a_{n+1}}{a_{n}}\bigg]Q_{-n}(x),
$$
from which we obtain (\ref{RecuRe1}). 
The relation (\ref{RecuRe2}) is obtained in the same manner. 
\end{proof}

\par
Note that the polynomials in the set
$\{Q_n(x)\}_{n=0,\pm 1,\pm 2,\ldots}$ 
can be ordered as follows
$$
Q_{0}(x), \hspace{2mm} Q_{1}(x), \hspace{2mm} Q_{-1}(x), \hspace{2mm} 
\cdots, \hspace{2mm} Q_{n}(x), \hspace{2mm} Q_{-n}(x), \cdots
$$ 
since $Q_{n}(x)$ and $Q_{-n}(x)$ have the same degree. 
This means that the relations (\ref{RecuRe1}) and (\ref{RecuRe2}) can be viewed as the three-term recurrence relations 
of $\{Q_n(x)\}_{n=0,\pm 1,\pm 2,\ldots}$. 

%
%


\section{Conclusion}
\par
Let us sum up. We introduced and characterized orthogonal functions that we have called Dunkl-supersymmetric. 
These functions are eigenfunctions of a class of Dunkl-type differential operators 
that can be cast within Supersymmetric Quantum Mechanics. 

\par
This investigation has built and expanded upon the analysis in \cite{SUSYQMR} where two examples had been studied. 
A significant feature of these orthogonal function families is that they do not involve polynomials of all degrees 
but are rather organized in pairs of polynomials both of the same degree 
(where the examples in terms of the Jacobi polynomials may be viewed as polynomials in another variable) . 
The connection with SUSY Quantum Mechanics has been exploited to obtain a number of Dunkl-SUSY orthogonal functions 
from exactly solvable problems. 
Informed by these results we could offer a general characterization of the Dunkl-SUSY OPs and 
could exhibit as well their recurrence relations. 
It would be of interest to relate the families of OPs that have been obtained as 
$q = -1$ limits of $q$-orthogonal polynomials \cite{XBI, DopBI, B-1J, L-1J}. 
A challenging project for the future would be to undertake the study of multivariate supersymmetric polynomials.

\section{Appendix}
\par

\noindent
{\bf Hermite polynomials} 
\begin{itemize}
\item 
The Hermite polynomials are orthogonal on the interval ($-\infty,\infty$) 
with respect to the weight function $e^{-x^2}$. 
They satisfy the orthogonality relations:
$$
\int^{\infty}_{-\infty}H_{m}(x)H_{n}(x)e^{-x^2}dx=2^n n!\sqrt{\pi}\delta_{mn}
$$
where $H_{n}(x)$ is the Hermite polynomial of degree $n$, 
$\delta_{mn}$ is Kronecher's delta. 
\item Three term recurrence relations: 
$$
H_{n}(x)=2xH_{n-1}(x)-2(n-1)H_{n-2}(x), \quad n\geq 1,
\quad H_{0}(x)=1.
$$
\end{itemize}

\noindent
{\bf Laguerre polynomials} 
\begin{itemize}
\item 
The Laguerre polynomials are orthogonal over $[0, \infty)$ 
with respect to the weight function $x^{\alpha}e^{-x}$. 
They satisfy the orthogonality relations:
$$
\int^{\infty}_{0}L^{(\alpha)}_{m}(x)L^{(\alpha)}_{n}(x)x^{\alpha}e^{-x}dx=\frac{\Gamma(n+\alpha+1)}{n!}\delta_{mn}
$$
where $L^{(\alpha)}_{n}(x)$ is the Laguerre polynomial of degree $n$. 
\item Three term recurrence relations: 
$$
L^{(\alpha)}_{n}(x)=\frac{2n-1+\alpha-x}{n}L^{(\alpha)}_{n-1}(x)-\frac{n-1+\alpha}{n}L^{(\alpha)}_{n-2}(x), \quad n\geq 1,
\quad L^{(\alpha)}_{0}(x)=1.
$$
\item Relation to Hermite polynomials: 
\begin{eqnarray}
\label{H2L_1}
H_{2n}(x) \hspace{-2.4mm}&=&\hspace{-2.4mm} (-1)^{n}2^{2n}n!L^{(-1/2)}_{n}(x^2),  \\
\label{H2L_2}
H_{2n+1}(x) \hspace{-2.4mm}&=&\hspace{-2.4mm} (-1)^{n}2^{2n+1}n!xL^{(1/2)}_{n}(x^2). 
\end{eqnarray}
\end{itemize}

\noindent
{\bf Jacobi polynomials} 
\begin{itemize}
\item 
The Jacobi polynomials are orthogonal on the interval ($-1,1$) 
with respect to the weight function $(1-x)^{\alpha}(1+x)^{\beta}$. 
They satisfy the orthogonality relations:
$$
\int^{\infty}_{-\infty}P^{(\alpha,\beta)}_{m}(x)P^{(\alpha,\beta)}_{n}(x)(1-x)^{\alpha}(1+x)^{\beta}dx=
\frac{2^{\alpha+\beta+1}\Gamma(\alpha+n+1)\Gamma(\beta+n+1)}
{n!\Gamma(\alpha+\beta+n+1)(\alpha+\beta+2n+1)}\delta_{mn}
$$
where $P^{(\alpha,\beta)}_{n}(x)$ is the Jacobi polynomial of degree $n$, 
and $\alpha,\beta>-1$.  
\item Three term recurrence relations: 
$$
P^{(\alpha,\beta)}_{n}(x)=\left[\frac{(2n+\alpha+\beta-1)(2n+\alpha+\beta)}{2n(n+\alpha+\beta)}x
-\frac{(\beta^2-\alpha^2)(2n+\alpha+\beta-1)}{2n(n+\alpha+\beta)(2n+\alpha+\beta-2)}\right]P^{(\alpha,\beta)}_{n-1}(x)
$$
$$\hspace{12mm}
-\frac{2(n+\alpha-1)(n+\beta-1)(2n+\alpha+\beta)}{2n(n+\alpha+\beta)(2n+\alpha+\beta-2)}P^{(\alpha,\beta)}_{n-2}(x), 
\quad n\geq 1
\quad P^{(\alpha,\beta)}_{0}(x)=1.
$$
\end{itemize}

\section*{Acknowledgements}
\par
We thank the referee for pointing out the two cases of QM models in our examples, 
which was addressed in Remark \ref{rmk_R}. 
%

\section*{References}

\end{document}